\documentclass[11pt]{article}

\usepackage[bibliosources=refs.bib]{pomegranate}

\DeclareOperator{\arg}
\DeclareOperator{\Arg}
\DeclareOperator{\poly}
\newcommand{\cA}{\mathcal A}
\newcommand{\ip}[2]{\dotprod{#1,#2}}
\newcommand{\GammaSector}[1]{\Gamma_{#1}}

\title{G{\aa}rding's Theorem for Posynomials}
\author[1]{Nima Anari}
\affil[1]{Stanford University, \texttt{anari@stanford.edu}}
\date{}

\begin{document}

\maketitle

\begin{abstract}
We extend G{\aa}rding's theorem to homogeneous posynomials: if a finite
positive sum of monomials with arbitrary nonnegative real exponents is
zero-free on a product of right half-planes, then its degree-normalized root
is concave.  Consequently, zero-freeness in a sector of aperture
\(\alpha\pi\) implies \(\alpha\)-fractional log-concavity.  This sharpens
generic mixing and domain-sparsification guarantees for fixed-size matchings
and nonsymmetric determinantal point processes.  The result was developed in
an AI-assisted interaction initiated and checked by the author; Codex also
assisted with assembling and typesetting the manuscript.
\end{abstract}

\section{Introduction}

G{\aa}rding's theorem turns a complex-analytic hypothesis into a convexity
statement: the degree-normalized root of a homogeneous hyperbolic polynomial
is concave on its hyperbolicity cone \cite{Garding1959}.  For a homogeneous
polynomial with nonnegative coefficients, zero-freeness on a product of right
half-planes therefore implies log-concavity on the positive orthant.  This
connection is now a standard tool in the study of negative dependence,
log-concave generating polynomials, and Markov-chain sampling
\cite{BorceaBrandenLiggett2009,AnariOveisGharanVinzant2018,BrandenHuh2019}.

We prove the same theorem for \emph{posynomials}.  A posynomial is a finite
sum
\[
  p(x)=\sum_{a\in\cA}c_a x^a
  =\sum_{a\in\cA}c_a\prod_{i=1}^n x_i^{a_i},
  \qquad c_a>0,\qquad a\in\R_{\ge0}^n.
\]
It is homogeneous of degree \(d\) if \(\sum_i a_i=d\) for every
\(a\in\cA\).  On
\(\C\setminus(-\infty,0]\), real powers are defined using the principal
logarithm; this convention defines a holomorphic branch of \(p\) on the
corresponding product domain.

For \(0<\alpha\le1\), let
\[
  \GammaSector{\alpha}
  =
  \set{re^{i\theta}\given r>0,\ \abs{\theta}<\alpha\pi/2}.
\]
Thus \(\GammaSector{1}\) is the open right half-plane.  We say that \(p\) is
\(\GammaSector{\alpha}\)-stable if its principal branch has no zero on
\(\GammaSector{\alpha}^n\).
We use \(\arg\parens{z}\in(-\pi,\pi)\) for the principal argument; when a
zero-free path specifies a continuous lift, we write \(\Arg\) instead.

\begin{theorem}[G{\aa}rding's theorem for posynomials]\label{thm:main}
Let \(p\) be a homogeneous posynomial of degree \(d>0\).  If \(p\) is
\(\GammaSector{1}\)-stable, then
\[
  x\longmapsto p(x)^{1/d}
\]
is concave on \(\R_{>0}^n\).  In particular, \(\log p\) is concave on
\(\R_{>0}^n\).
\end{theorem}

The half-plane statement contains the sector statement by a power map.

\begin{corollary}[Sharp fractional log-concavity]\label{cor:sector}
Let \(0<\alpha\le1\), and let \(p\) be a homogeneous posynomial of degree
\(d>0\).  If \(p\) is \(\GammaSector{\alpha}\)-stable, then
\[
  x\longmapsto p(x^\alpha)^{1/(\alpha d)}
\]
is concave on \(\R_{>0}^n\).  Consequently,
\(x\mapsto\log p(x^\alpha)\) is concave.
\end{corollary}

The dependence on \(\alpha\) is best possible.  Indeed,
\(p(x,y)=x^\beta+y^\beta\) is \(\GammaSector{\alpha}\)-stable exactly when
\(\alpha\beta\le1\), and
\((x^{\alpha\beta}+y^{\alpha\beta})^{1/(\alpha\beta)}\) is concave on the
positive orthant exactly in the same range.

We consider real exponents in order to sharpen the known connection between
sector-stability and fractional log-concavity.  A homogeneous polynomial
\(f\) is \(\alpha\)-fractionally log-concave if
\[
  x\longmapsto\log f(x^\alpha)
\]
is concave on the positive orthant.  Earlier work showed that
\(\GammaSector{2\alpha}\)-stability implies
\(\alpha\)-fractional log-concavity
\cite[Lemma~69]{AlimohammadiAnariShiragurVuong2021}.  Thus a sector of
aperture \(\beta\pi\) yielded only \(\beta/2\)-fractional log-concavity.
\Cref{cor:sector} removes this factor-of-two loss.

This loss propagated into the polynomial exponents in subsequent sampling
results.  Removing it yields quadratic improvements in the resulting
mixing-time and domain-sparsification guarantees, including those for
fixed-size matchings and nonsymmetric determinantal point processes
\cite{AnariJainKoehlerPhamVuong2021a,AnariDerezinskiVuongYang2022}.

\section{Ray Monotonicity for Real Exponents}\label{sec:ray}

We first record the real-exponent form of Descartes' rule.  For a finite
sequence of nonzero real numbers \(b_0,\ldots,b_N\), let
\(V(b_0,\ldots,b_N)\) denote its number of sign changes.

\begin{lemma}[Generalized Descartes rule]\label{lem:descartes}
Let \(\lambda_0<\cdots<\lambda_N\) be real numbers and let
\[
  F(t)=\sum_{j=0}^N b_jt^{\lambda_j},
  \qquad b_j\in\R\setminus\set{0}.
\]
The number of zeros of \(F\) on \(\R_{>0}\), counted with multiplicity, is at
most \(V(b_0,\ldots,b_N)\).
\end{lemma}

\begin{proof}
We argue by induction on the number of terms.  Multiplication by
\(t^{-\lambda_0}\) does not change the positive zeros, so write
\[
  G(t)=b_0+\sum_{j=1}^N b_jt^{\mu_j},
  \qquad 0<\mu_1<\cdots<\mu_N.
\]
The derivative factors as
\[
  G'(t)=t^{\mu_1-1}H(t),
  \qquad
  H(t)=\sum_{j=1}^N b_j\mu_jt^{\mu_j-\mu_1}.
\]
By induction, the number of positive zeros of \(G'\) is at most
\(V(b_1,\ldots,b_N)\).

Let \(Z(G)\) count the positive zeros of \(G\) with multiplicity.  Rolle's
theorem, including the multiplicities at the zeros themselves, gives
\(Z(G)\le Z(G')+1\).  If \(b_0\) and \(b_1\) have the same sign, this can be
improved to \(Z(G)\le Z(G')\).  To see this, take \(\delta>0\) small enough
that both \(G(\delta)\) and \(G'(\delta)\) have that common sign.  If \(x\) is
the first positive zero of \(G\), the mean value theorem gives a point in
\((\delta,x)\) where \(G'\) has the opposite sign.  Thus \(G'\) has one zero
before \(x\), in addition to the zeros supplied by the usual Rolle count
between and at the zeros of \(G\).

If \(b_0b_1>0\), then
\(V(b_0,\ldots,b_N)=V(b_1,\ldots,b_N)\); if \(b_0b_1<0\), the former is one
larger.  The two Rolle bounds prove the claim.
\end{proof}

The next theorem is the positive-coefficient ray theorem of Sendov and
Sendov \cite[Lemma~3.3(ii)]{SendovSendov2021}, with
\cref{lem:descartes} in place of the ordinary Descartes rule.
For completeness, we give their crossing argument in full.  The only
extension is that the exponents below need not be integers.

\begin{theorem}[Sendov--Sendov ray monotonicity]\label{thm:ray}
Let
\[
  f(u)=\sum_{j=0}^N a_ju^{\lambda_j},
  \qquad a_j>0,\qquad 0\le\lambda_0<\cdots<\lambda_N,
\]
where powers are defined using the principal logarithm.  Fix
\(0<\theta<\pi\), and suppose that
\(f(\rho e^{i\theta})\ne0\) for every \(\rho>0\).  Then the continuous lift
\(\varphi\) of the argument along this ray, normalized by
\[
  \lim_{\rho\downarrow0}\varphi(\rho)=\lambda_0\theta,
\]
is nondecreasing.  In particular,
\[
  \lambda_0\theta
  \le \varphi(\rho)\le
  \lambda_N\theta
  \qquad(\rho>0).
\]
\end{theorem}

\begin{proof}
The asymptotics of the first and last terms give
\[
  \lim_{\rho\downarrow0}\varphi(\rho)=\lambda_0\theta,
  \qquad
  \lim_{\rho\to\infty}\varphi(\rho)=\lambda_N\theta.
\]
Fix \(\beta\in\R\) such that no \(\lambda_j\theta\) is congruent to
\(\beta\) modulo \(\pi\), and consider
\[
  H_\beta(\rho)
  =\Im\parens*{e^{-i\beta}f(\rho e^{i\theta})}
  =\sum_{j=0}^N a_j\rho^{\lambda_j}
    \sin(\lambda_j\theta-\beta).
\]
Let
\[
  N_\beta
  =\#\set*{\ell\in\beta+\pi\Z\given
    \lambda_0\theta<\ell<\lambda_N\theta}.
\]
As \(j\) increases, the sign of
\(\sin(\lambda_j\theta-\beta)\) can change only when
\(\lambda_j\theta\) crosses one of these \(N_\beta\) levels.  Hence the
coefficient sequence of \(H_\beta\) has at most \(N_\beta\) sign changes.
By \cref{lem:descartes}, \(H_\beta\) has at most \(N_\beta\) positive zeros.

The endpoint values of \(\varphi\) give the reverse count.  For every
\(\ell\in\beta+\pi\Z\) strictly between the two endpoints, continuity gives
some \(\rho>0\) with \(\varphi(\rho)=\ell\), and then
\(H_\beta(\rho)=0\).  Distinct levels give distinct points because \(f\) has
no zero on the ray.  Thus \(H_\beta\) has exactly \(N_\beta\) positive zeros,
one for each level between the endpoints.

Since \(\beta\) was arbitrary, this exact-count conclusion holds for every
admissible \(\beta\).

Suppose that \(\varphi\) is not nondecreasing.  Choose \(0<s<t\) such that
\(\varphi(s)>\varphi(t)\).  We may choose
\[
  \varphi(t)<\ell<\varphi(s)
\]
so that \(\ell\) is not an endpoint and is not congruent modulo \(\pi\) to
any \(\lambda_j\theta\).  Set \(\beta=\ell\).  Then \(\beta\) is admissible,
and the path makes a downward passage through the level \(\ell\) between
\(s\) and \(t\).

If \(\ell\) lies between the endpoint values, the path starts below it and
ends above it.  It therefore meets \(\ell\) at least three times: once before
\(s\), once between \(s\) and \(t\), and once after \(t\).  Together with
one crossing of each of the other \(N_\beta-1\) levels, this gives at least
\(N_\beta+2\) positive zeros of \(H_\beta\).  If \(\ell\) lies outside the
endpoint interval, it is not counted by \(N_\beta\), but the excursion gives
two crossings of \(\ell\) in addition to the \(N_\beta\) crossings between
the endpoints.  Again \(H_\beta\) has at least \(N_\beta+2\) positive zeros,
contradicting \cref{lem:descartes}.
\end{proof}

\section{Proof of the Main Theorem}\label{sec:proof}

Ray monotonicity gives the multivariate estimate that drives the proof.

\begin{lemma}[Angular contraction]\label{lem:angular}
Let \(p\) be a homogeneous \(\GammaSector{1}\)-stable posynomial of degree
\(d\).  Suppose
\[
  w_i=r_ie^{i\theta_i},
  \qquad r_i>0,\qquad 0\le\theta_i\le\phi<\pi.
\]
Then \(p(w)\ne0\).  Moreover, the continuous argument obtained from the path
\(s\mapsto(r_ie^{is\theta_i})_{i=1}^n\), \(0\le s\le1\), satisfies
\[
  0\le\Arg\parens{p(w)}\le d\phi.
\]
\end{lemma}

\begin{proof}
The claim is immediate when \(\phi=0\), so assume \(\phi>0\) and set
\(\tau_i=\theta_i/\phi\).  Define the one-variable posynomial
\[
  f(u)=p(r_1u^{\tau_1},\ldots,r_nu^{\tau_n})
  =\sum_{a\in\cA}c_ar^a u^{\ip{a}{\tau}}.
\]
Terms with equal exponents may be combined, so \(f\) has positive
coefficients and strictly increasing real exponents.  Homogeneity gives
\[
  0\le\ip{a}{\tau}\le d
  \qquad(a\in\cA).
\]
Here we used \(\tau_i\in[0,1]\), which ensures that principal logarithms give
\((r_iu^{\tau_i})^{a_i}=r_i^{a_i}u^{a_i\tau_i}\) throughout
\(\C\setminus(-\infty,0]\).

We claim that \(f\) is zero-free on
\(\C\setminus(-\infty,0]\).  Write \(u=\rho e^{i\omega}\) with
\(\abs{\omega}<\pi\).  The arguments of the points
\(r_iu^{\tau_i}\) lie between \(0\) and \(\omega\), an interval of width
less than \(\pi\).  After multiplying all coordinates by
\(e^{-i\omega/2}\), they lie in the open right half-plane.  The principal
logarithms do not cross their branch cut during this rotation, and
homogeneity gives
\[
  p(r_1u^{\tau_1},\ldots,r_nu^{\tau_n})
  =e^{id\omega/2}
    p(e^{-i\omega/2}r_1u^{\tau_1},\ldots,
      e^{-i\omega/2}r_nu^{\tau_n}).
\]
The last factor is nonzero by stability, proving the claim.

Apply \cref{thm:ray} to \(f\) on the ray of angle \(\phi\).  Its smallest
and largest exponents lie in \([0,d]\), so
\[
  0\le\Arg\parens{f(e^{i\phi})}\le d\phi.
\]
Since \(f\) is zero-free on \(\C\setminus(-\infty,0]\) and positive on the
positive real axis, the lift in \cref{thm:ray} is the same branch obtained by
moving from \(1\) to \(e^{i\phi}\) in the upper half-plane.  Finally,
\(f(e^{i\phi})=p(w)\), and this path is precisely the one in the statement.
\end{proof}

We use the following standard consequence of the Pick representation.  It is
included to make clear the precise analytic input.

\begin{lemma}[Pick concavity]\label{lem:pick}
Let \(h\) be holomorphic on \(\C\setminus(-\infty,0]\), positive on
\(\R_{>0}\), and suppose that
\[
  0\le\arg\parens{h(z)}\le\arg\parens{z}
  \qquad\text{when }\Im\parens{z}>0.
\]
Then \(h\) is concave on \(\R_{>0}\).
\end{lemma}

\begin{proof}
The hypotheses say that \(h\) is a complete Bernstein function.  Its Pick
representation \cite[Theorem~6.2]{SchillingSongVondracek2012} has the form
\[
  h(x)=a+bx+\int_{[0,\infty)}\frac{x}{x+t}\,d\mu(t),
  \qquad a,b\ge0,
\]
for a positive measure \(\mu\) satisfying the usual integrability condition.
For every \(t\ge0\), the function \(x\mapsto x/(x+t)\) is concave on
\(\R_{>0}\).  The displayed representation therefore proves the claim.
\end{proof}

\begin{proof}[Proof of \cref{thm:main}]
Fix \(x,y\in\R_{>0}^n\), and define
\[
  g(z)=p(zx+y),
  \qquad z\in\C\setminus(-\infty,0].
\]
Each coordinate \(zx_i+y_i\) avoids the branch cut on this domain, so \(g\)
is holomorphic there.  If \(\Im\parens{z}>0\), then
\[
  0\le\arg\parens{zx_i+y_i}\le\arg\parens{z}<\pi
  \qquad(1\le i\le n).
\]
By \cref{lem:angular}, \(g(z)\ne0\), and the branch of the argument that is
zero on the positive real axis satisfies
\[
  0\le\Arg\parens{g(z)}\le d\,\arg\parens{z}.
\]
Indeed, the angular lift in \cref{lem:angular} varies continuously with
\(z\) and is zero on the positive real axis, so it is the branch normalized
there.
The corresponding statement in the lower half-plane follows by conjugation.
Thus \(g\) is zero-free on \(\C\setminus(-\infty,0]\).

Let \(L\) be the holomorphic logarithm of \(g\) which is real on
\(\R_{>0}\), and set
\[
  h(z)=\exp\parens*{L(z)/d}.
\]
Then \(h\) is positive on \(\R_{>0}\), and angular contraction gives
\[
  0\le\arg\parens{h(z)}\le\arg\parens{z}
  \qquad\text{for }\Im\parens{z}>0.
\]
By \cref{lem:pick}, \(h\) is concave on \(\R_{>0}\).

Define the \emph{perspective} of \(h\) by
\[
  H(s,t)=t h(s/t),
  \qquad s,t>0.
\]
We verify that \(H\) is concave.  Fix \((s_1,t_1),(s_2,t_2)\in\R_{>0}^2\)
and \(0\le\lambda\le1\), and set
\[
  T=\lambda t_1+(1-\lambda)t_2,
  \qquad
  \eta=\frac{\lambda t_1}{T}.
\]
Then
\[
  \frac{\lambda s_1+(1-\lambda)s_2}{T}
  =
  \eta\frac{s_1}{t_1}+(1-\eta)\frac{s_2}{t_2}.
\]
Concavity of \(h\) therefore gives
\begin{align*}
  H\parens*{\lambda s_1+(1-\lambda)s_2,T}
  &\ge
  T\parens*{\eta h(s_1/t_1)+(1-\eta)h(s_2/t_2)} \\
  &=
  \lambda H(s_1,t_1)+(1-\lambda)H(s_2,t_2).
\end{align*}
Thus \(H\) is concave.  Homogeneity of \(p\) gives
\[
  H(s,t)=p(sx+ty)^{1/d}.
\]
The perspective extends continuously to \(\R_{\ge0}^2\), so its concavity also
holds on the boundary.  Comparing the points \((1,0)\) and \((0,1)\) proves
\[
  p(\lambda x+(1-\lambda)y)^{1/d}
  \ge
  \lambda p(x)^{1/d}+(1-\lambda)p(y)^{1/d}
\]
for \(0\le\lambda\le1\).  Hence \(p^{1/d}\) is concave.  Since the logarithm
is increasing and concave, \(\log p=d\log(p^{1/d})\) is concave as well.
\end{proof}

\begin{proof}[Proof of \cref{cor:sector}]
Set \(q(z)=p(z^\alpha)\), with powers taken coordinatewise.  This is a
homogeneous posynomial of degree \(\alpha d\).  If
\(z_i\in\GammaSector{1}\), then
\(z_i^\alpha\in\GammaSector{\alpha}\), so \(q\) is
\(\GammaSector{1}\)-stable.  Applying \cref{thm:main} to \(q\) gives the
claim.
\end{proof}

\section{Consequences for Sampling}\label{sec:consequences}

For a distribution \(\mu\) on \(\binom{[n]}{k}\), write
\[
  g_\mu(z)=\sum_{S\in\binom{[n]}{k}}\mu(S)z^S.
\]
By \cref{cor:sector}, \(\GammaSector{\alpha}\)-stability of \(g_\mu\)
implies \(\alpha\)-fractional log-concavity.  Sector-stability is preserved
by nonzero links and positive external fields, so the same conclusion holds
for every link under every external field.  We now combine this observation
with existing sampling theorems.

For \(\ell\le k\), the \(k\leftrightarrow\ell\) down-up walk first chooses a
uniform \(\ell\)-subset of its current \(k\)-set and then resamples a
containing \(k\)-set according to \(\mu\).

\begin{corollary}[Mixing from a sector-stability certificate]\label{cor:mixing}
Suppose that \(g_\mu\) is \(\GammaSector{\alpha}\)-stable, and let
\(r=\ceil*{1/\alpha}\).  If \(k\ge r\), the
\(k\leftrightarrow k-r\) down-up walk has modified log-Sobolev constant
\[
  \Omega_\alpha\parens{k^{-1/\alpha}}.
\]
Consequently, from a state \(S_0\) in the support of \(\mu\), its total
variation mixing time satisfies
\[
  t_{\mathrm{mix}}(\varepsilon)
  =O_\alpha\parens*{
    k^{1/\alpha}
    \parens*{\log\parens*{1+\log\frac1{\mu(S_0)}}
      +\log\frac1\varepsilon}
  }.
\]
\end{corollary}

This follows from the entropic-independence theorem for fractionally
log-concave distributions
\cite[Theorem~5]{AnariJainKoehlerPhamVuong2021a}.

\begin{corollary}[Domain sparsification]\label{cor:sparsification}
Suppose that \(g_\mu\) is \(\GammaSector{\alpha}\)-stable and that \(\mu\)
satisfies the oracle and marginal-estimation assumptions of the
domain-sparsification framework.  Then sampling from \(\mu\) reduces to
sampling from distributions obtained by sparse external fields supported on
\[
  n^{1-\alpha}\poly(k)
\]
elements.
\end{corollary}

This is the domain-sparsification theorem of Anari, Derezi{\'n}ski, Vuong,
and Yang \cite{AnariDerezinskiVuongYang2022}, with the improved fractional
log-concavity parameter supplied by \cref{cor:sector}.

\paragraph{Applications.}
The generating polynomials for fixed-size matchings and nonsymmetric
\(k\)-DPPs are \(\GammaSector{1/2}\)-stable
\cite[Lemmas~9 and~12]{AlimohammadiAnariShiragurVuong2021}; the matching
certificate builds on the Heilmann--Lieb theorem \cite{HeilmannLieb1972}.
For fixed-size matchings, duality gives the same certificate for the
matched-vertex encoding.  Thus \cref{cor:mixing} gives two-site down-up walks
with modified log-Sobolev constant \(\Omega(k^{-2})\), where \(k\) denotes
the degree of the corresponding homogeneous encoding.  Under the assumptions
of \cref{cor:sparsification}, both families reduce to sparse domains of size
\(n^{1/2}\poly(k)\), improving the \(n^{3/4}\poly(k)\) bound obtained from
the earlier fractional log-concavity parameter.

For the standard homogeneous encodings of the full monomer-dimer model and
the full nonsymmetric DPP, the same two-site entropy bounds were already
obtained in \cite{AnariJainKoehlerPhamVuong2021a} from separate proofs of
\(1/2\)-fractional log-concavity.  Here they, as well as the
fixed-cardinality statements, follow uniformly from the sector-stability certificate.

\section*{Acknowledgments and AI Use}

We thank Misha Ivkov for discussions about the problem.  We also thank Amire
Bendjeddou, Thu Hien Nguyen, Greg Knese, and Kevin Shu for discussions at the
2025 AIM workshop \emph{The Geometry of Polynomials in Combinatorics and
Sampling}.  We thank AIM for its hospitality.

The author formulated the theorem and its connections to sector-stability,
fractional log-concavity, entropic independence, and sampling.  A GPT~5.5 Pro
Extended interaction suggested the angular-contraction route.  In subsequent
work, Codex (GPT~5.6 Sol) extracted and formalized the real-exponent extension
of the Sendov--Sendov crossing argument used to complete the posynomial proof.
Codex also assisted with assembling and typesetting the manuscript.  The
author checked and revised the proof and is responsible for all statements
and arguments in the final manuscript.

\PrintBibliography

\end{document}